\documentclass[letterpaper]{article}
\usepackage[numbers]{natbib}
\usepackage{geometry}
\usepackage{amsmath, amsthm, bm}
\usepackage{amssymb}
\usepackage{ctable} 
\usepackage{caption}

\usepackage{hyperref}
\usepackage[english]{babel}
\usepackage[utf8]{inputenc}

\usepackage{color}

\newtheorem{theorem}{Theorem}

\usepackage{array,booktabs,ragged2e}
\newcolumntype{R}[1]{>{\RaggedLeft\arraybackslash}p{#1}}
\usepackage{stackengine}
\usepackage{scalerel}
\usepackage[linesnumbered,ruled,vlined]{algorithm2e}

\usepackage{multirow}

\usepackage[ruled,vlined]{algorithm2e}
\usepackage{graphicx}
\graphicspath{ {./images/} }
\usepackage[affil-it]{authblk}
\author[1]{Di Bo}
\author[1,*]{Hoon Hwangbo}
\author[1]{Vinit Sharma}
\author[1]{Corey Arndt}
\author[1]{Stephanie C.~TerMaath}

\affil[1]{University of Tennessee, Knoxville}
\affil[*]{Corresponding author: Hoon Hwangbo, hhwangb1@utk.edu}

\title{A Randomized Subspace-based Approach for Dimensionality Reduction and Important Variable Selection}

\date{}

\begin{document}

\maketitle

\begin{abstract}
An analysis of high-dimensional data can offer a detailed description of a system but is often challenged by the curse of dimensionality.  General dimensionality reduction techniques can alleviate such difficulty by extracting a few important features, but they are limited due to the lack of interpretability and connectivity to actual decision making associated with each physical variable. Variable selection techniques, as an alternative, can maintain the interpretability, but they often involve a greedy search that is susceptible to failure in capturing important interactions or a metaheuristic search that requires extensive computations. This research proposes a new method that produces subspaces, reduced-dimensional physical spaces, based on a randomized search and leverages an ensemble of critical subspace-based models, achieving dimensionality reduction and variable selection.  When applied to high-dimensional data collected from the failure prediction of a composite/metal hybrid structure exhibiting complex progressive damage failure under loading, the proposed method outperforms the existing and potential alternatives in prediction and important variable selection.

\end{abstract}

\smallskip
\noindent \textit{keywords:}  Subspace-based modeling, randomized algorithms, feature selection, hybrid material analysis, damage tolerance modeling

\section{Introduction}

With the recent breakthrough in computing, a real-world data analysis tends to involve voluminous high-dimensional datasets for modeling various systems with great complexity \citep{reddy2020analysis, zhu2020high}. 
An increased dimensionality can provide more information about an underlying system, but analyzing a high-dimensional dataset is not an easy task due to the curse of dimensionality \citep{liu2019variable}.  That is, for a fixed sample size, data tend to be sparser in spaces characterized by  high-dimensional as compared to low-dimensional parameters, so signals are much weaker but noises have far greater impacts, often leading to improper analysis outcomes.
A common solution to this problem is to evaluate a data-driven model in a meaningful low-dimensional space by identifying a subset of features that can represent the entire data with a minimal loss of information, e.g., through dimensionality reduction or variable selection \citep{liu2019variable}.

Dimensionality reduction extracts inherent features from high-dimensional data and relocates the data in a low-dimensional space constructed by the features, alleviating the curse of dimensionality \citep{van2009dimensionality}. 
However, existing dimensionality reduction techniques are limited due to the lack of interpretability as the extracted features are ``artificial''  having obscure connections to the original ``physical'' variables \citep{van2009dimensionality}. 
In engineering applications, identifying important physical variables is crucial for reducing the number of occasions for operational controls or experimental designs used to optimize a system of interest. 
Although dimensionality reduction can offer low-dimensional evaluation of data, its limited capability restricts its broader usage for system optimization.  As an alternative, variable selection, also referred to as subset selection or feature selection, can distinguish important physical variables by extracting a subset of the original variables that are significant \citep{wei2015variable}.  Finding an exact optimal subset is computationally intractable, so variable selection typically involves a heuristic search for the subset \citep{fong2013selecting}.  A simple heuristic such as a greedy approach cannot ensure the accuracy of the final reduced model, and a complex heuristic such as genetic algorithm is computationally expensive.  In addition, while applying a heuristic search, important interactions between physical variables can easily be missed.  

This paper develops 
a randomized subspace-based modeling for the purpose of alleviating the curse of dimensionality and providing valuable insight about an underlying system for system optimization.  To this end, we model a regression problem as an ensemble of multiple base learners where each base learner is responsible for representing a lower-dimensional input space only; we will refer to this reduced-dimensional physical space as a subspace.  
The subspace-based modeling is similar to general variable selection approaches as it finds certain subsets of the original variables to achieve dimensionality reduction but differs from them as it leverages multiple distinct subsets and accordingly multiple models to evaluate a function.  To balance between model accuracy and computational complexity, we randomly generate subspaces and selectively choose important subspaces. By construction, an important subspace informs that not only the variables therein but also the potential interactions between them are significant.  We prescribe the subspaces to be a fixed (low) dimension considering that a high-order interaction exceeding a certain order is rarely significant in practice.  By keeping all subspaces at a low dimension, functional evaluation is always performed at a low-dimensional space without a risk of suffering from the curse of dimensionality.

The potential of this method to significantly impact our characterization and understanding of advanced engineering systems is demonstrated using the challenging problem of identifying the most influential material properties on damage tolerance for a layered hybrid structure and formulating a reduced order model based on these most sensitive parameters.  This example was chosen due to the high dimensional parameter space and wide range of parameter values. Additionally, the high fidelity model used to predict damage tolerance requires a prohibitive amount of computational time to characterize just a small number of parameters in a narrow subspace of the parameter value ranges. The approach presented herein enables a novel method to rapidly reduce and characterize this vast and complex parameter space to solve a challenging physics-based structural mechanics problem.


The remainder of this article is organized as follows. Section~\ref{sc:lit_review} reviews relevant studies on dimensionality reduction and important variable selection. Section~\ref{sc:subs_model} presents the proposed subspace-based method by describing model structure, randomized generation and selective extraction of subspaces, and overall learning process. Section~\ref{sc:result} demonstrates the benefits of the proposed method in comparison with others for modeling the damage tolerance of a hybrid material and discusses its potential usages for structural design, analysis, and optimization. Section~\ref{sc:conc} concludes the paper and discusses future work.

\section{Literature Review}
\label{sc:lit_review}
In the past decades, dimensionality reduction has been common in many applications involving a large number of variables, including digital photographs, speech recognition, and bioinformatics \citep{van2009dimensionality}. Dimensionality reduction techniques transform high-dimensional data into a meaningful reduced-dimensional representation, ideally close to its intrinsic dimensions \citep{van2009dimensionality}. The intrinsic dimensions of data are the minimum features needed to account for the observed properties of the data \citep{fukunaga2013introduction}.   

Traditional dimensionality reduction techniques include 
Principal Component Analysis (PCA), Independent Component Analysis (ICA), and Multidimensional Scaling (MDS). PCA is one of the most popular linear reduction techniques and dates back to Karl Pearson in 1901 \citep{pearson1901liii}. This technique tries to find orthogonal directions that account for as much variance of data as possible.  
Due to its attractive advantages of minimal information loss and generation of uncorrelated dimensions, PCA is still popular for use in a broad range of application areas \citep{zou2006sparse}.  For example, \citet{li2016evaluating} applied PCA to evaluate energy security of multiple East Asian countries with respect to vulnerability, efficiency, and sustainability.  \citet{salo2019dimensionality} proposed a network anomaly detection method based on a reduced dimensionality achieved by combining information gain and PCA. 
On the other hand, ICA tries to extract independent pieces of information from high-dimensional data, mainly for the purpose of source blind separation that has been popular in signal processing \citep{hastie2009elements}. ICA can be useful in other areas of study; for example, \citet{sompairac2019independent} discussed the benefits of ICA to unravel the complexity of cancer biology, particularly in analyzing different types of omics datasets. 
Different from PCA and ICA, MDS is a nonlinear dimensionality reduction technique. It provides a useful graphical representation of data based on the similarity information of individual data points. MDS is a common technique for analyzing network-structured data as presented in \citet{saeed2019state} that applied MDS to wireless networks localization. 

Over the recent decades, many nonlinear or nonparametric dimensionality reduction techniques have been proposed \citep{lee2007nonlinear, saul2006spectral}. Kernel PCA is a reformulation of traditional linear PCA constructing feature space through a kernel function \citep{scholkopf1998nonlinear}. \citet{choi2005fault} found PCA was inefficient and problematic for modeling a nonlinear system but kernel PCA effectively captured nonlinearity while achieving dimensionality reduction. 
\citet{xu2019software} proposed a defect prediction framework that combined kernel PCA and weighted extreme learning machine to extract representative data features and learn an effective defect prediction model.
\citet{tenenbaum2000global} proposed an isometric feature mapping (Isomap) technique that was based on classical MDS but sought to retain the intrinsic geometry of datasets as captured in the geodesic manifold distances. 
It achieves the goal of nonlinear dimensionality reduction by using easily measured local metric information to learn the underlying global geometry of a data set.
\citet{hinton2006reducing} suggested using deep autoencoder networks while transforming high-dimensional data into low-dimensional codes by training a multilayer neural network with a small central layer.  They presented that the deep autoencoder networks outperformed PCA with a proper selection of initial weights.  \citet{belkina2019automated} introduced $t$-distributed stochastic neighbor embedding (t-SNE), and \citet{gisbrecht2015parametric} presented kernel t-SNE as an extension of t-SNE to a parametric framework, which enabled explicit out-of-sample extensions. They demonstrated that kernel t-SNE yielded satisfactory results for large data sets. 
\citet{mcinnes2018umap} proposed uniform manifold approximation and projection (UMAP) constructed by a theoretical framework based on Riemannian geometry and algebraic topology. Compared to t-SNE, UMAP arguably preserves more of the global structure. Since UMAP does not have computational restrictions on embedding dimension, it can be easily used for general dimensionality reduction.

Even with the theoretical and methodological advance of the dimensionality reduction techniques, they still cannot identify important physical variables among the existing variables; instead, they extract inherent features that are artificial. 
However, determining important variables is a critical task in many experimental studies, including those for structural mechanics \citep{termaath2018probabilistic}, environmental science, and bioinformatics \citep{wei2015variable}.  This is because knowing important variables can suggest which variables to test and optimize for subsequent experiments to reduce the number of experimental trials and hence expedite the overall experimental procedure.  In this context, variable selection (or feature selection) that extracts a subset of existing variables is a good alternative to typical dimensionality reduction. It is capable of maintaining the original variable structure while reducing the dimensionality. 

Feature selection methods are typically classified into three groups of filter methods, wrapper methods, and embedded methods \citep{chandrashekar2014survey}.  Filter methods evaluate each variable based on some ranking criteria, such as Pearson correlation coefficient \citep{battiti1994using} or mutual information \citep{torkkola2003feature}.  Wrapper methods, evaluating different subsets of variables by using a learning algorithm, generally provide better performance compared to filter methods \citep{xue2015survey}.  To determine the subsets of variables subject to evaluation, an exhaustive search (rarely used) requires considering $2^n$ different feature combinations when there are $n$ variables, which is impractical.  Instead, sequential forward selection \citep{whitney1971direct} and sequential backward selection \citep{marill1963effectiveness} algorithms that apply greedy search have been used broadly.  More recently, sequential floating forward selection \citep{pudil1994floating} and adaptive sequential floating forward selection \citep{somol1999adaptive} have been proposed to alleviate poor performance of the greedy search.  However, these sequential selection methods still suffer from a nesting problem caused by the nature of the greedy search, so the selection outcome is generally far from the optimum.  To overcome the nesting problem, heuristic search methods, e.g., based on genetic algorithm \citep{alexandridis2005two}, have been proposed.  These wrapper methods, however, require an extensive amount of computation and are prone to overfitting \citep{chandrashekar2014survey}.  Embedded methods aim to alleviate the computational cost of wrapper methods by integrating feature selection and model learning into a single process \citep{xue2015survey}.  In other words, some ranking criteria, such as max-relevancy min-redundancy \citep{peng2005feature} and the weights of a classifier \citep{mundra2009svm}, are used for the initial feature reduction, and then wrapper methods are applied for the final feature selection.  

These days, machine learning methods, such as random forest (RF) \citep{chen2020selecting} and neural network (NN) \citep{olden2004accurate}, have been used to measure variable importance and rank variables based on model's prediction accuracy. 
\citet{gregorutti2017correlation} explored the usage of RF in the presence of correlated predictors. Their results motivated the use of the recursive feature elimination (RFE) algorithm using permutation importance measure as a ranking criterion. \citet{chen2020selecting} compared results from three popular datasets (bank marketing, car evaluation database, human activity recognition using smartphones) with and without feature selection based on multiple method, including RF and RFE. 
The experimental results demonstrated that RF achieved the best performance in all experimental groups. 
On the other hand, \citet{olden2004accurate} compared nine variants of artificial neural network (ANN) for quantifying variable importance using simulated data. For this comparison, the true importance of variables was known and used for verification. Their results showed connection weight approach (CWA) provided the best accuracy and CWA was the only method that correctly identified the rank of variable importance. 
\citet{liu2019variable} presented a deep Baysian rectified linear unit (ReLU) network for high-dimensional regression and variable selection. The result showed their method outperformed existing classic and other NN-based variable selection methods.  
Although these variable selection methods provide an efficient tool to extract important variables, they primarily focus on the importance of individual variables, with little consideration for the importance of feature interactions.  In a general experimental problem, however, knowledge about the significance of feature interactions is a critical factor for designing experiments.

Instead of finding important individual variables, our approach aims to identify critical subspaces, each of which presents which variable combination as a whole (including interactions) is important.  For the exploration of potentially important subspaces, a randomized search is employed to improve model accuracy and computational efficiency relative to a greedy search and a metaheuristic search, respectively.   Furthermore, our method leverages multiple subsets of variables (subspaces) for model construction instead of relying on a single subset as the existing methods do, which can provide a more flexible model.  The details of the proposed method will be discussed in the subsequent sections.

\section{Subspace-Based Modeling and Critical Subspaces}
\label{sc:subs_model}
In this section, we develop randomized subspace-based modeling for solving a general supervised learning problem and identifying critical variables and interactions.  The dataset of interest contains data pairs of $\{\mathbf{x}_i, y_i\}_{i=1}^n$ where $\mathbf{x}_i$ is a $p$-dimensional input vector and $y_i$ is the corresponding response. For subspace-based modeling, we form a subspace, a space spanned by a subvector of input $\mathbf{x}$, determine the significance of a subspace for modeling response, and use such a critical subspace as a basic unit for model building.  A critical subspace implies that the variables forming the space and their interactions are important, so it naturally extracts important variables and interactions.  In the following sections, we describe the proposed model structure, the randomized search for generating subspaces, the significance evaluation of a subspace, and the overall model learning process.

\subsection{Subspace-based model}
Considering a general supervised learning problem, we estimate a function $f:\mathbb{R}^p \rightarrow \mathbb{R}$ that relates $p$-dimensional input $\mathbf{x}$ and output $y$ as $y=f(\mathbf{x})+\varepsilon$ where $\varepsilon$ is an additive noise. 
We model the unknown function $f$ as an additive mixture of subfunctions $g_j:\mathbb{R}^k \rightarrow \mathbb{R}$ for $j=1,\ldots, J$ defined in a lower dimensional space of dimension $k\ll p$, instead of estimating the full-dimensional function directly, as
\begin{equation}
y = f(\mathbf{x}) + \varepsilon \approx g_1(\mathbf{z}_1) + g_2(\mathbf{z}_2) + ... + g_J(\mathbf{z}_J) + \varepsilon,
\label{eq:model}
\end{equation}
where $\mathbf{z}_j$ is the $j$th subvector of $\mathbf{x}$ of size $k$.  Each $\mathbf{z}_j$ for $j=1,\ldots,J$ takes different components of $\mathbf{x}$.  For example, suppose $k=3$ and $\mathbf{x} = \begin{bmatrix}
x_1, x_2, \ldots, x_p
\end{bmatrix}^T
$ where $p > 20$.
Then, we may have $\mathbf{z}_1 = \begin{bmatrix}
x_3, x_8, x_{12}
\end{bmatrix}^T$ and $\mathbf{z}_2 = \begin{bmatrix}
x_7, x_{11}, x_{20}
\end{bmatrix}^T$.
We define a reduced-dimensional space spanned by each subvector $\mathbf{z}_j$ as a subspace, and $g_j$ estimates the response $y$ within a subspace formed by $\mathbf{z}_j$.  As such, in Eq.~\eqref{eq:model}, we model the function $f$ as a mixture of subspace-based models.

The advantage of the subspace-based modeling is obvious. By evaluating $y$ in low-dimensional subspaces, there is no risk of suffering from the curse of dimensionality, which is not true when evaluating $y$ in a single high-dimensional space.  Different from general dimensionality reduction methods, the subspace-based modeling does not require extracting artificial dimensions, but the reduction of dimensionality is achieved by forming low-dimensional physical spaces, which can provide physical interpretation of the model.  One major shortcoming of the subspace-based modeling is that it is not capable of modeling interactions of an order higher than $k$. However, in many real-world problems, a high-order interaction is almost negligible in modeling response.  From a preliminary study, we found that $k=3$, i.e., modeling up to 3-factor interactions, provided good prediction results.

The key to the success in the subspace-based modeling lies in how to form the subspaces, i.e., how to generate the subvectors of $\mathbf{z}_j$ for $j=1,\ldots,J$.  For an exhaustive search, if we assume $p=41$ and $k=3$, the number of all subspace candidates is ${41 \choose 3} = 10660$.  Evaluating models for this many subspaces and determining whether to include each of them requires a considerable amount of computation (still, much smaller compared to $2^{41}$ for an exhaustive search for the optimal subset).  On the other hand, simple feature selection techniques, such as greedy forward selection and greedy backward selection, are limited in terms of the exploration capability.  In the next section, we discuss how to generate subspaces and how to extract critical subspaces.

\subsection{Subspace generation and extraction}

To explore a broad area covering potential subspace configurations and not to rely on a greedy search, we generate subspaces randomly.  In particular, at each draw, we randomly choose $k$ variables out of $p$ variables (full-dimension) and evaluate whether to include this randomly generated subspace into the model or not.  For each random draw, we apply sampling without replacement, but we allow duplicated selection of a single variable at multiple draws.  In other words, an input variable $x_1$ could be a part of subspace $\mathbf{z}_1$, and the variable can be chosen again at later draws for $\mathbf{z}_j$ for $j>1$.  This is to ensure that multiple interactions associated with a key variable are not lost and all of them can be used for modeling the response as long as they have significant impacts.

We determine the significance of a randomly generated subspace by evaluating the percentage reduction in the prediction error with respect to root mean square error (RMSE).  To evaluate out-of-sample prediction error and avoid possible overfitting, we apply 5-fold cross validation (CV) to calculate the error.  
This 5-fold CV is applied to the dataset assigned for model learning (excluding testing portion), and the dataset is split into 5 small subsets of the data.  Each subset can serve as a validation dataset whereas the remaining four subsets collectively can be used to train a model.  In this way, we can generate 5 different train/validation combinations of data, for each of which train and validation portions account for 80\% and 20\% of the dataset, respectively, as shown in Fig.~\ref{fig:extraction}. 
\begin{figure}[b!]
\centering
\includegraphics[width=\textwidth]{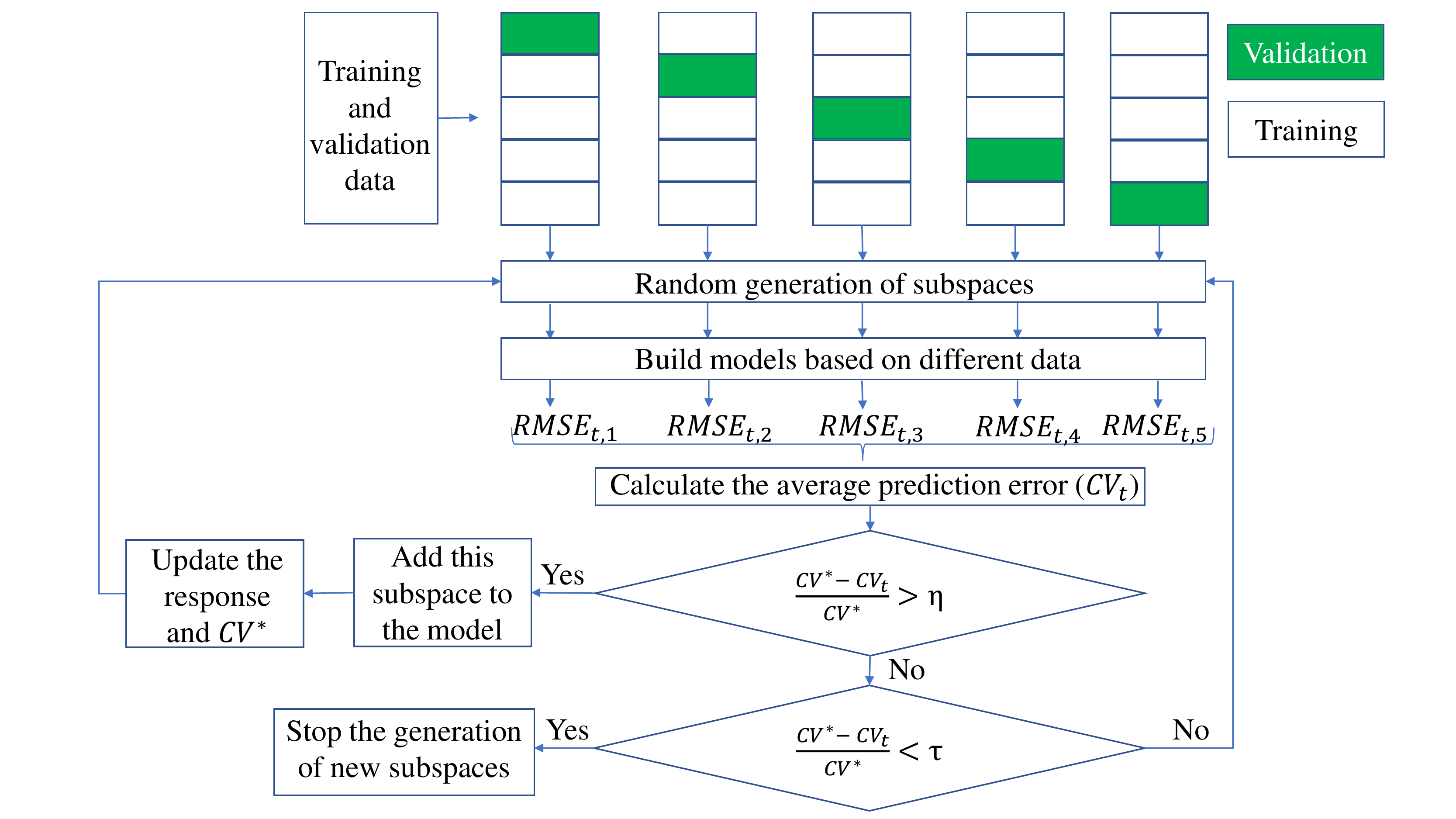}
\caption{A flow chart describing the process of subspace generation and extraction }
\label{fig:extraction}
\end{figure}

To test the significance, a temporary subspace-based model ($g_t$) for a randomly generated subspace ($\mathbf{z}_t$) will be added into the model, and the prediction error of this extended model will be evaluated in terms of the out-of-sample RMSE.  The temporary subspace-based model, $g_t$ where $t$ denotes the current iteration number for the subspace generation, is learned for each of the five different train datasets, and the prediction error of the extended model including $g_t$ is calculated by using the corresponding five validation datasets.  This will produce $RMSE_{t,q}$ for $q=1,\ldots,5$ for each train/validation split at iteration $t$. To be specific, the RMSE is calculated as
\begin{equation}
    RMSE_{t,q} = \sqrt{\frac{1}{n_{v_q}}\sum_{\{\mathbf{x}_i, y_i\} \in \mathcal{V}_q}{(y_i - \tilde{f}_t(\mathbf{x}_i))^2}}
\end{equation}
where $\mathcal{V}_q$ is the validation dataset for the $q$th train/validation split and $n_{v_q}$ is the number of observations in the validation dataset.  $\tilde{f}_t(\mathbf{x}_i) = \sum_{j \in \mathcal{J}^*} g_j(\mathbf{z}_j) + g_t(\mathbf{z}_t)$ where $\mathcal{J}^*$ is the set of indices for the selected critical subspaces until the previous iteration.
The CV score is then the average of the five prediction errors, i.e.,
\begin{equation}
 CV_t =\frac{1}{5} \sum_{q=1}^{5}  RMSE_{t,q}.
 \label{eq:cv}
\end{equation}

We use the percentage reduction of the prediction error for the selection of a critical subspace and also for the termination of the iterative search for new subspaces.  The percentage error reduction is defined as
\begin{equation}
\Delta e = \frac{CV^* - CV_t}{CV^*} 
\label{eq:epr}
\end{equation}
where $CV^*$ is the minimum (best) CV score in Eq.~\eqref{eq:cv} obtained until the previous iteration.
The initial value of $CV^*$ is calculated from a model that only includes a constant term (the average of the five sample means of $y_i$'s in each training dataset).
For the selection of a critical subspace, we set a selection threshold $\eta$, and if \(\Delta e>\eta\), the subspace generated at the $t$th iteration is accepted and included in the model.  This means that a subspace that reduces the prediction error at least by a certain percentage will be included in the model.  Otherwise, this subspace is disregarded.  
The search for a new subspace continues until the percentage reduction \(\Delta e\) becomes less than a termination threshold, \(\tau\).  In other words, if \(\Delta e<\tau\), the subspace searching process terminates.  If none of the selection criterion and the termination criterion are met, the iterative search continues for the next possible critical subspace. 
Fig.~\ref{fig:extraction} illustrates the overall process of the subspace generation and extraction.

\subsection{Subspace model learning and hyperparameter selection} \label{ssc:hyppar}
This section describes how to estimate each subspace-based model $g_t$.  For this estimation, we use a randomly generated subvector, $\mathbf{z}_t$ as predictors and the current residual calculated before the iteration $t$ as a response to estimate, i.e., using $\tilde{y}_{t,i} = y_i - \sum_{j \in \mathcal{J}^*} g_j(\mathbf{z}_j)$ for $i=1,\ldots,n$ for response. In other words, we consider a regression problem written as
\begin{equation}
    \tilde{y}_{t,i} = g_t(\mathbf{z}_{t,i}) + \tilde{\varepsilon}_{t,i}
    \label{eq:resid_model}
\end{equation}
where $\mathbf{z}_{t,i}$ is a subvector of $\mathbf{x}_{i}$, and $\tilde{\varepsilon}_{t,i}$ is a modified noise for $i=1,\ldots,n$.

To learn the function $g_t$, we use support vector machine (SVM) as a base learner. It is a supervised learning method which can be used for classification and regression. For a regression analysis, it is also known as support vector regression (SVR).  To train a model based on SVR, we use ``$\epsilon$-insensitive'' error function which is a symmetric loss function. 
For this function, a flexible tube with the radius of $\epsilon$ is formed symmetrically around the estimate.  Errors of size less than \(\epsilon\) are ignored, and overestimation and underestimation are penalized equally.  This means that some penalty is assigned only to the points outside the tube based on the distance between the points and the tube, but no penalty is applied to the points within the tube \citep{awad2015support}. 

For SVR, the function $g_t$ can be expressed as a linear combination of basis functions, $h_m$ for $m=1,\ldots,M$, as
\begin{equation}
    g_t(\mathbf{z}_t) = \sum_{m=1}^M{\beta_m h_m(\mathbf{z}_t)} + \beta_0,
\end{equation}
where $\beta_m$ for $m=1,\ldots,M$ is a coefficient of each basis function and $\beta_0$ is a constant.  To estimate $g_t$, we minimize
\begin{equation}
\begin{split}
    H(\bm{\beta}, \beta_0) &= \sum_{i=1}^n{V_\epsilon(\tilde{y}_{t,i} - \hat{g}_t(\mathbf{z}_{t,i}))} + \frac{\lambda}{2}\sum_{m=1}^M{\beta_m^2}  \\
    &\text{where} \; V_\epsilon(r) = \begin{cases}
    0, & \text{if} \; |r|<\epsilon, \\
    |r| - \epsilon, & \text{otherwise}.
    \end{cases}
\end{split} \label{eq:obj}
\end{equation}
where $\bm{\beta} = \begin{bmatrix}
\beta_1, \ldots, \beta_M
\end{bmatrix}^T$ is a coefficient vector, and $\lambda$ is a regularization parameter.  The minimizer of Eq.~\eqref{eq:obj} provides the estimate of $g_t$ in the form of
\begin{equation}
    \hat{g}_t(\mathbf{z}) = \sum_{i=1}^n{\alpha_i}K_t(\mathbf{z}, \mathbf{z}_{t,i}). 
\end{equation}
where $\alpha_i$ for $i=1,\ldots,n$ is a Lagrangian dual variable for the Lagrangian primal shown in Eq.~\eqref{eq:obj}, and $K_t(\cdot,\cdot)$ is a kernel function that represents the inner product of the unknown basis functions, $h_m$ for $m=1,\ldots,M$ \citep{hastie2009elements}.


To fully specify the estimate of $g_t$, some hyperparameters need to be determined.  This includes the regularization parameter, $\lambda$ (related to the cost parameter in a typical SVM), the radius of the tube, $\epsilon$, for the loss function, and some kernel related parameters.  For this hyperparameter learning, we apply grid search based on generalized cross-validation (GCV) criterion. The proposed method already employs 5-fold CV for the selection of critical subspaces and the termination of the iterative search, so another layer of out-of-sample prediction will complicate the data structure, and the training process may suffer from the lack of a sufficient number of data points.  While measuring in-sample error using training data only, GCV provides a convenient approximation to the leave-one-out CV \citep{hastie2009elements}, so it is a proper criterion for the hyperparameter learning of the proposed method. In general, GCV is calculated as
\begin{equation}
GCV(\hat{f}) = \frac{1}{n}\sum_{i=1}^{n} \Big[\frac{y_i- \hat{f}(\mathbf{x}_i)}{1-\text{trace}(\mathbf{S})/n)}\Big]^2
\label{eq:gcv}
\end{equation}
where $\mathbf{S}$ is an $n\times n$ hat matrix that performs a linear projection of $\mathbf{y}=\begin{bmatrix}
y_1, \ldots, y_n
\end{bmatrix}^T$ to achieve the estimate of $\mathbf{y}$, i.e., $\hat{\mathbf{y}} = \mathbf{S y}$. 
For the proposed subspace-based model in Eq.~\eqref{eq:model}, deriving this expression of linear projection is not straightforward.  Theorem~\ref{thm:trace} shows a good approximation of $\mathbf{S}$ that can be derived under the assumption of a squared loss function.
\begin{theorem}
Assuming a squared loss function, i.e., $V_\epsilon(r) = r^2$, the estimate of the response $\mathbf{y}$ can be expressed as $\hat{\mathbf{y}} = \mathbf{S y} = (\sum_{j=1}^J{\mathbf{S}_j})\mathbf{y}$ where $\mathbf{S}_j = \mathbf{S}_j^\prime (\mathbf{I} - \sum_{l=0}^{j-1}{\mathbf{S}_l})$, $\mathbf{S}_j^\prime = \mathbf{K}_j(\mathbf{K}_j + \lambda \mathbf{I})^{-1}$, $\mathbf{S}_0 = \mathbf{0}$, $\mathbf{I}$ is an $n\times n$ identity matrix, and $\mathbf{K}_j$ is an $n\times n$ kernel matrix with $\{\mathbf{K}_j\}_{i, i^\prime} = K_j(\mathbf{z}_{j,i}, \mathbf{z}_{j,i^\prime})$ for $i, i^\prime =1,\ldots,n$.
\label{thm:trace}
\end{theorem}

\begin{proof}
From Eq.~\eqref{eq:model}, we define $\mathbf{S}_j \mathbf{y}$ for $j=1,\ldots,J$ to represent $g_j(\mathbf{z}_{j,i})$ as a linear combination of $y_i$'s so that we have
\begin{equation}
    \hat{\mathbf{y}} = \mathbf{S}_1 \mathbf{y} + \mathbf{S}_2 \mathbf{y} + \cdots + \mathbf{S}_J \mathbf{y}.
\end{equation}
Since each $g_j$ estimates residuals from a model including up to the previous subspace-based model ($g_{j-1}$), we have $\tilde{y}_{j,i} = g_j(\mathbf{z}_{j,i}) + \tilde{\varepsilon}_{j,i}$ (as in Eq.~\eqref{eq:resid_model}).  For an SVR with a squared loss function, the estimation can be expressed as $\widehat{\tilde{\mathbf{y}}}_j=\mathbf{S}_j^\prime \tilde{\mathbf{y}}_j = \mathbf{K}_j (\mathbf{K}_j + \lambda \mathbf{I})^{-1} \tilde{\mathbf{y}}_j$ (see \citet{hastie2009elements}) where $\widehat{\tilde{\mathbf{y}}}_j$ is the estimate of $\tilde{\mathbf{y}}_j$, a vector including $\tilde{y}_{j,i}$ for $\forall i$, and $\{\mathbf{K}_j\}_{i, i^\prime} = K_j(\mathbf{z}_{j,i}, \mathbf{z}_{j,i^\prime})$ for $i, i^\prime =1,\ldots,n$ for a given kernel $\mathbf{K}_j$.  Since $\tilde{\mathbf{y}}_j = \mathbf{y} - \mathbf{S}_1 \mathbf{y} - \mathbf{S}_2 \mathbf{y} - \cdots - \mathbf{S}_{j-1} \mathbf{y}$,
\begin{equation}
    \mathbf{S}_j \mathbf{y} := \widehat{\tilde{\mathbf{y}}}_j=\mathbf{S}_j^\prime \tilde{\mathbf{y}}_j = \mathbf{S}_j^\prime (\mathbf{y} - \mathbf{S}_1 \mathbf{y} - \mathbf{S}_2 \mathbf{y} - \cdots - \mathbf{S}_{j-1} \mathbf{y}) = \mathbf{S}_j^\prime (\mathbf{I} - \sum_{l=1}^{j-1}{\mathbf{S}_l})\mathbf{y} = \mathbf{S}_j^\prime (\mathbf{I} - \sum_{l=0}^{j-1}{\mathbf{S}_l})\mathbf{y}.
\end{equation}
Now, for $j=1$, the response to estimate is $\tilde{\mathbf{y}}_1 = \mathbf{y}$, so we have
\begin{equation}
    \mathbf{S}_1 \mathbf{y} := \widehat{\tilde{\mathbf{y}}}_1=\mathbf{S}_1^\prime \tilde{\mathbf{y}}_j = \mathbf{S}_1^\prime \mathbf{y} = \mathbf{S}_1^\prime (\mathbf{I} - {\mathbf{S}_0})\mathbf{y}.
\end{equation}
where $\mathbf{S}_0$ is an $n\times n$ zero matrix.  Therefore, the result holds for $\forall j=1,\ldots,J$.
\end{proof}
As implied in Theorem~\ref{thm:trace}, we do not allow distinct hyperparameters for each $g_j$ estimation as the GCV calculation is based on the full model.  Instead, we assume the hyperparameters to be the same across all $g_j$ for $j=1,\ldots,J$.  We apply a grid search for the hyperparameter learning, which is to keep hyperparameters at certain levels, add critical subspaces based on this hyperparameter setting, and evaluate the GCV value for the resulting full model (after the termination of the search). The set of hyperparameters and the corresponding subspace-based model that minimize the GCV criterion in Eq.~\eqref{eq:gcv} will be chosen as an optimal model.

Because the GCV calculation based on the result of Theorem~\ref{thm:trace} is complicated due to the recursive format, we also consider a simpler version of the trace calculation.  For approximation, we use $\mathbf{S}_j^\prime$ instead of $\mathbf{S}_j$.  The following describes two alternatives we propose for the GCV calculation:
\begin{enumerate}
    \item[\textbf{A1}.] $\text{trace}(\mathbf{S}) = \sum_{j=1}^J{\text{trace}(\mathbf{S}_j)}$ based on Theorem~\ref{thm:trace}.
    \item[\textbf{A2}.] $\text{trace}(\mathbf{S}) = \sum_{j=1}^J{\text{trace}(\mathbf{S}_j^\prime)}$ for simplification.
\end{enumerate}
The performance of the two alternatives will be compared and discussed in Section~\ref{sc:result}.



\subsection{Overall algorithm} 
\begin{algorithm}[t]
\SetAlgoLined
Apply 5-fold CV to split given data into training data ($\mathcal{T}_q$) and validation data ($\mathcal{V}_q$)\;
Build a constant model by using data in $\mathcal{T}_q$ and predict responses for the corresponding $\mathcal{V}_q$\;
Initialize $CV^*$ by setting it to the prediction error of the constant model (the average of 5 errors obtained from $\mathcal{V}_q$ for $q=1,\ldots,5$)\;
Create a grid for assessing different sets of hyperparameters\;

\For {each level combination of hyperparameters}{
Set the iteration number to 0, i.e., $t=0$\;
\Repeat{$\Delta e < \tau$}{
  $t\leftarrow t+1$\;
  Randomly draw a subspace\;
  Use $\mathcal{T}_q$ and build an SVR model by using the randomly drawn subspace\;
  Make a prediction of $\tilde{y}_{t,i}$ for those in the validation datasets, $\mathcal{V}_q$\;
  Calculate the prediction error as $\tilde{y}_{t,i} - \hat{g}_t(\mathbf{z}_{t,i})$ and accordingly $RMSE_{t,q}$ and $CV_t$\;
  \If{$\Delta e > \eta$}{
    $CV^* \leftarrow CV_t$\;
    $\mathcal{J}^* \leftarrow \mathcal{J}^* \cup \{t \}$\;
  }
}
Use the entire model learning data ($\mathcal{T}_q \cup \mathcal{V}_q$) to build a full model as in Eq.~\eqref{eq:model} by using the set of subspaces, $\mathcal{J}^*$, determined above\;
Calculate GCV (either A1 or A2)\;
}
Find the minimum GCV\;
Determine the optimal hyperpameters and finalize the full model specification\;
 \caption{Learning process of the randomized subspace modeling}
 \label{alg:overall}
\end{algorithm}

For clear illustration of the entire learning process, Alg.~\ref{alg:overall} shows the step-by-step procedures of the proposed subspace-based modeling.  To build a model, we use a dataset that is dedicated to model learning; if testing is needed (as in Section~\ref{sc:result}), we split the original dataset into two distinct datasets, one for model learning and another for testing.  From this model learning dataset, we apply 5-fold CV to split training data and validation data where validation data are required to determine the significance of subspaces and the convergence of the algorithm.  After the data split, we produce a model that includes a single constant term by using training data only and calculate the error of validation set prediction to initialize $CV^*$.  Each hyperparameter setting will be evaluated once a full model is established, i.e., once all critical subspaces are determined and added into the model.  For each fixed level set of hyperparameters, we generate a subspace randomly and build an SVR model to estimate $g_t$.  The quality of a subspace is evaluated via 5-fold CV with respect to $CV_t$ and $\Delta e$ in Eq.~\eqref{eq:cv} and \eqref{eq:epr}.  If a subspace meets the selection criterion, it will be added into the model.  This process of generating and evaluating a random subspace will continue until the termination criterion is met.  Once the algorithm terminates, the GCV value will be computed according to Eq.~\eqref{eq:gcv}.  After completing all the iterations for the grid search, we find the optimal hyperparameter values and apply this optimal setting to the entire model learning dataset (training and validation) to build a final model.  This final model will be evaluated by a separate test set for the comparison with other methods in Section~\ref{sc:result}.

\section{Case Study}
\label{sc:result}
Design  and  analysis  of  layered  hybrid  structure  is  challenged  by  the  many  possible  choices  of materials and configurations.  This vast parameter space is impractical to explore through physical testing and is computationally prohibitive due to the analysis time required due to the large number and ranges of input parameters.  The randomized subspace approach is applied to overcome this limitation and provide an efficient and accurate approach to computationally characterize the parameter space and inform limited physical testing to define parameters. The reduced order model can then be used to rapidly explore the parameter space to customize and optimize layered designs. The specific problem of metal and composite layered structures was chosen due to its complexity and fit for the demonstration objectives.
 
To demonstrate the effectiveness of the proposed method, we compare its performance with other possible alternatives.  The methods subject to comparison include various wrapper methods and machine learning-based methods.  Some methods using a metaheuristic search are not considered here due to their high computational requirement.  Since all alternatives leverage a learning model for variable selection, we first compare the prediction accuracy of models obtained from the alternatives.  We also compare important variables selected by the alternatives.  To compare the results more effectively, 5-fold CV is used here; in specific, it is used for the comparison of prediction ability and the evaluation of important variable selection. The prediction error is measured by RMSE, and the average and standard deviation of 5 RMSE's are compared. Since the other methods focus on the importance of ``individual'' variables, we extract a common set of important variables obtained from the 5-fold CV-based implementation (even for the proposed method).  

\subsection{Finite element analysis}
Numerical simulation using a validated finite element model (finite element analysis) is an efficient method to design and predict the damage tolerance of a layered hybrid structure for varying material properties.  However, given the nearly unlimited choices of material combinations and stacking sequences, it is not possible to explore every possible design and optimize for all possible material and configuration parameters. Therefore, identification of the
most influential parameters is necessary to limit the design and optimization parameter space to obtain a computationally tractable solution. The case study layered hybrid configuration is an aluminum plate with a co-cured bonded quasi-isotropic E-glass/epoxy composite overlay. The composite overlay consists of 8 layers of multiple lamina types (see Table \ref{table:stacking sequence}) resulting in a high dimensional number of material parameters needed for characterization and input into the finite element model. 

\begin{table}[t!]
\caption{Composite patch stacking sequence} 
\centering 
\begin{tabular}{c c c} 
\toprule
 & E-glass fabric & Fabrication style \\ 
\midrule 
1 & Hexcel 7781 & 0$^{\circ}$/90$^{\circ}$ Stain weave\\ 
2 & Vectorply E-BX 1200 & $\pm$ 45$^{\circ}$ Stitch  \\
3 & Vectorply E-LT 1800 & 0$^{\circ}$/90$^{\circ}$ Stitch \\
4 & Vectorply E-BX 1200 & $\pm$ 45$^{\circ}$ Stitch  \\
5 & Vectorply E-BX 1200 & $\pm$ 45$^{\circ}$ Stitch  \\ 
6&Vectorply E-LT 1800 & 0$^{\circ}$/90$^{\circ}$ Stitch \\
7& Vectorply E-BX 1200 & $\pm$ 45$^{\circ}$ Stitch  \\
8& Hexcel 7500 & 0$^{\circ}$/90$^{\circ}$ Plain weave\\

\bottomrule
\end{tabular}
\label{table:stacking sequence}
\end{table}

This layered multi-material model was loaded under four point bending to engage progressive damage in the structure through multiple damage mechanisms.  Damage
tolerance was evaluated by the total energy absorbed by the structure (an output parameter
that is calculated during analysis and readily extracted using an automated approach). This
model captures multiple, interacting damage mechanisms including the plastic deformation
in aluminum, shear plasticity in each lamina, the intralaminar fracture of each lamina,
delamination within the patch and disbond at the interface.
 Evaluation using this total damage energy provides a distinct and measurable result for the development of a reduced order predictive model and evaluation of influential parameters and their interactions.

A 3D high fidelity finite element model explicitly captures each layer in the hybrid structure as well as the interfaces \citep{heng2018prediction}.
Each fabric layer (lamina) is explicitly modeled, and cohesive elements are included between each layer to capture delamination between plies. Each lamina is individually modeled with continuum shell elements (SC8R). A cohesive damage model is implemented for each lamina using a VUMAT user subroutine. Cohesive elements with a triangular traction-separation law are used to detect the interlaminar damage and are also included at the metal/composite interface to capture disbond between the metal and resin. The aluminum substrate is modeled with solid elements (C3D8R). Loading and support pins are modeled as rigid bodies to create the boundary and loading conditions. The numerical simulations are executed in the FE code ABAQUS. This physics-based model was validated under four point bend loading through physical testing (Fig.~\ref{fig:validation}) .

\begin{figure}[h!]
\centering
\includegraphics[width=.6\textwidth]{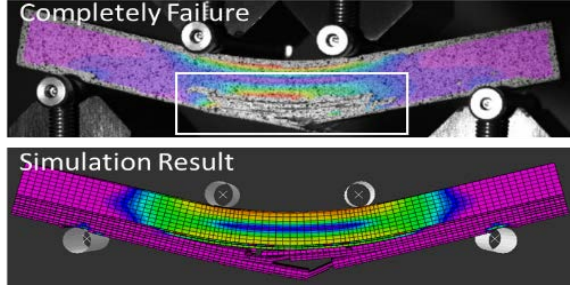}
\caption{Comparison of experimental (top) and numerical results (bottom) at failure under four point bending}
\label{fig:validation}
\end{figure}

To predict the total energy absorption of the layered hybrid structure, the 41 parameters characterizing the material properties needed for input into the finite element model for the multiple layers
are used as the predictors (see Table~\ref{table:realvariable} for detail). Latin hypercube sampling is applied to sample the parameter space based on the mean and standard deviation of the parameter ranges. For this case study, completing a single analysis generating a single data record takes an average of 3 hours depending on the parameter combination. Due to the
computational time required to analyze the finite element model, we conducted
200 analyses producing a predictor matrix of size 200$\times$41 and a response vector
of size 200.  With more analyses, it is possible to generate a larger dataset, and this can further improve the prediction quality.  However, from a practical perspective, we aim to develop a method that works well even with a small dataset, so we compare alternatives based on this small dataset.

\begin{table}[!htbp]
\caption{Variables for modeling damage tolerance of a hybrid metal structure (the Resin in this table is ``Resin (M1002 with M2046 Hardener)'') \citep{heng2018prediction} } 
\centering 
\resizebox{\textwidth}{!}{
\begin{tabular}{lll} 
\toprule 
Variable &Corresponding parameter \\
\midrule
$x_1$ ($\sigma_y$) &Yield Stress of Al-5456  \\ 
$x_2$ ($n$) & Strain Hardening Exponent of Al-5456 \\
$x_3$ ($\sigma_{ss}$)&Nominal Stress First/Second Direction of Resin between Laminate plies  \\
$x_4$ ($E_{Al}$)& Young's Modulus  of Al-5456    \\
$x_5$ ($P$)& Power Term in Shear Hardening Equation for Laminates  \\ 
$x_6$ ($\sigma_{nn}$)& Nominal Stress Normal-only Mode of Resin between Laminate plies  \\
$x_7$ ($X_{7781}$)& Tensile strength of the Laminae Reinforced with Hexcel 7781\\
$x_8$ ($G_{1200}$)& Intralaminar Fracture Toughness of Laminae Reinforced with EBX 1200  \\ 
$x_9$ ($\epsilon^{pl}_{max}$)& Maximum Shear Plastic Strain for Laminates  \\
$x_{10}$ ($G_{II}$)& Shear Mode Fracture Energy First/Second Direction of Resin  between Laminate plies  \\
$x_{11}$ ($\alpha_{12}$)& Shear Damage Parameter for Laminates\\
$x_{12}$ ($E_{1800}$) & Young’s Modulus of Laminae Reinforced with ELT 1800 \\
$x_{13}$ ($G_I$)& Normal Mode Fracture Energy of Resin between Laminate plies\\
$x_{14}$ ($X_{7500}$)& Tensile strength of the Laminae Reinforced with Hexcel 7500   \\
$x_{15}$ ($\sigma_{nni}$)& Nominal Stress Normal-only Mode of Resin between metal/composite interface \\
$x_{16}$ ($v_{Al}$)& Poisson’s Ratio  of Al-5456     \\
$x_{17}$ ($E_{7500}$)& Young’s Modulus of Laminae Reinforced with Hexcel 7500 \\
$x_{18}$ ($\sigma_{ssi}$)& Nominal Stress First/Second Direction of Resin between metal/composite interface \\
$x_{19}$ ($X_{1800}$)& Tensile strength of the Laminae Reinforced with ELT 1800    \\
$x_{20}$ ($B-K_i$)& Mixed Mode Behavior for Benzeggagh-Kenane of Resin \\
& between metal/composite interface \\
$x_{21}$ ($E_{1200}$)& Young’s Modulus of Laminae Reinforced with EBX 1200 \\
$x_{22}$ ($G_{1800}$)& Intralaminar Fracture Toughness of Laminae Reinforced with ELT 1800    \\
$x_{23}$ ($\tilde{\sigma}_y$)& Effective Shear Yield Stress  for Laminates\\
$x_{24}$ ($X_{12}$)&Shear Strength of Resin  for all Lamina \\
$x_{25}$ ($B$)& Strength Coefficient of Al-5456  \\
$x_{26}$ ($v_{7500}$)& Poisson Ratio of Laminae Reinforced with Hexcel 7500 \\ 
$x_{27}$ ($X_{1200}$)& Tensile strength of the Laminae Reinforced with EBX 1200 \\
$x_{28}$ ($v_{1200}$)& Poisson Ratio of Laminae Reinforced with EBX 1200  \\
$x_{29}$ ($v_{1800}$)& Poisson Ratio of Laminae Reinforced with ELT 1800  \\
$x_{30}$ ($G_{7500}$)& Intralaminar Fracture Toughness of Laminae Reinforced with Hexcel 7500  \\
$x_{31}$ ($E_{7781}$)& Young’s Modulus of Laminae Reinforced with Hexcel 7781  \\
$x_{32}$ ($v_{7781}$)& Poisson Ratio of Laminae Reinforced with Hexcel 7781  \\
$x_{33}$ ($G_{7781}$)& Intralaminar Fracture Toughness of Laminae Reinforced with Hexcel 7781  \\
$x_{34}$ ($G_{12}$)& Shear Modulus of Laminate for Laminates  \\
$x_{35}$ ($d^{max}_{12}$)& Maximum Shear Damage for Laminates \\
$x_{36}$ ($C$)& Coefficient in Shear Hardening Equation for Laminates \\
$x_{37}$ ($E_{nn}$)& Elastic Modulus of Resin between Laminate plies  \\
$x_{38}$ ($B-K$)& Mixed Mode Behavior for Benzeggagh-Kenane of Resin between Laminate plies  \\
$x_{39}$ ($E_{nni}$)& Elastic Modulus of Resin between metal/composite interface  \\
$x_{40}$ ($G_{Ii}$)& Normal Mode Fracture Energy of Resin between metal/composite interface  \\
$x_{41}$ ($G_{IIi}$)& Shear Mode Fracture Energy First/Second Direction of \\
&Resin  between metal/composite interface  \\
\bottomrule
\end{tabular}}
\label{table:realvariable}
\end{table}

\subsection{Implementation detail}
By design, the randomized subspace modeling includes a 5-fold CV for the selection of critical subspaces.  To compare performance with other methods, another layer of a 5-fold CV is implemented to have separate test datasets.  The similar treatment is applied to other methods being compared.  One layer of a 5-fold CV is used to split test datasets, and another layer of a 5-fold CV is used to learn hyperparameters, if applicable.  The other methods subject to comparison include regression models based on linear regression, lasso regression, ridge regression, principal component regression (PCR), partial least squares regression (PLS), RF, $k$-nearest neighbors ($k$-NN), SVR, and NN to cover various linear and nonlinear modeling commonly used for feature selection.  Before applying any method, we normalize the data to prevent any scale-relevant errors.

To implement the randomized subspace modeling, we set the dimension of subspaces to three, i.e., $k=3$, and the selection and termination thresholds to \(\eta = 1\%\) and \(\tau= 0.001\%\), respectively, based on the results of preliminary studies.  For the base learner of the subspace-based modeling (SVR), a kernel needs to be set among multiple possible options, including linear, polynomial (POLY), radial basis function (RBF), and sigmoid kernels.  From preliminary experiments, we found kernels capturing linearity, such as linear and POLY kernels, provided better performance.  In this regard, we choose POLY kernel, \(K(\mathbf{u}, \mathbf{v}) = (\gamma \mathbf{u}^T \mathbf{v} + \delta)^{d}\) where $\gamma$, $\delta$, and $d$, respectively, denote the scale, offset, and degree of polynomials, and set $\gamma=1/3$, $\delta=0$, and $d=1$ (found optimal). Since $\gamma$ had little impact on the final prediction error, we used the default value available in an R function for SVR fitting. Considering that we use normalized dataset, $\delta=0$ providing the optimal performance is not a surprise.  Although $d=1$ produces the optimal result, the POLY kernel was chosen over the linear kernel as it is more flexible and capable of modeling nonlinearity with different parameterization.
On the other hand, SVR learning itself involves some parameters, $\epsilon$ and $C := 1/\lambda$.  We test three levels for each of the two parameters, i.e., $\epsilon = 0.01, 0.1, 0.5$ and $C= 1, 2, 5$, which generates nine distinct combinations of parameter setting among which an optimal setting is determined.  When applying a grid search for this hyperparameter learning, if there exists a poor level combination, it is possible that the naive termination criterion does not work properly.  To prevent iterating more than what is required for the exhaustive search, we force the search process to stop after 10000 iterations.  For this case study, this hard thresholding was active only for a single level combination out of nine, so the termination criterion generally works well.

Excluding NN and the proposed method, all methods are executed in R using \texttt{train()} function in the \texttt{caret} package.  This package contains effective tools for data splitting, pre-processing, feature selection, model tuning with resampling and variable importance estimation as well as other functionality \citep{chen2020selecting}.  \citet{chen2020selecting} showed RF method with \texttt{varImp()} was an efficient tool for calculating variable importance, so we use the same function to measure variable importance for other methods.  On the other hand, NN is implemented in Python using  ``keras'' package.  To make the result comparable to the existing variable selection methods based on NN \citep{olden2004accurate,liu2019variable}, we test over multiple NN structures, one or two hidden layer(s), and three activation functions of linear, sigmoid, and ReLU.  We also evaluate various hyperparameter options including the numbers of neurons, batch size, and epochs.  With all this variation, we pick an NN model with the lowest prediction error for the result comparison.  For NN, variable importance is determined by the CWA used in \citet{olden2004accurate}.  Since all the methods but the randomized subspace modeling calculate variable importance rather than picking important variables, for these methods, we extract the first twenty important variables (based on the importance score) from each of the five-fold learning and find the common variables from the list of all selected variables to report critical individual variables.

\subsection{Comparison results}\label{ssc:result}
The optimal hyperparameter values of the proposed method vary a little with each model learning set (five of such).  Table~\ref{tab:opt_par} shows the optimal values selected based on the two GCV criterion suggested in Section~\ref{ssc:hyppar}.  From the table, it is shown, in general, $\epsilon=0.01$ and $C=5$ are found optimal.
\begin{table}[h!]
    \centering
    \caption{Optimal parameters selected for each of the five model learning/test data split}
    \begin{tabular}{lcccccc}
    \toprule
         && Learn/test 1 & Learn/test 2
         & Learn/test 3 & Learn/test 4
         & Learn/test 5
         \\
         \midrule
         \multirow{2}{*}{A1} & $\epsilon$ &0.01 &0.01 &0.01 &0.01 &0.01  \\
         & $C$ &5 &5&5&5&2\\
         \midrule
         \multirow{2}{*}{A2} & $\epsilon$ &0.1&0.01&0.01&0.01&0.01 \\
         & $C$ &5 &5&5&5&2\\
         \bottomrule
    \end{tabular}
    \label{tab:opt_par}
\end{table}

The optimal models based on the result in Table~\ref{tab:opt_par} are compared with other alternatives as shown in Table~\ref{table:linear}.  It is noticed that, in general, linear models perform better than nonlinear models for the material damage tolerance prediction.  This applies to both the average and standard deviation of the prediction errors.  Among the linear methods, partial least squares regression produces the lowest average prediction error while the regularized methods (lasso and ridge) show comparably decent performance.  Although linear regression has the smallest variance of the estimator, its average prediction error is relatively higher than the other methods.  For nonlinear methods, only SVM with a polynomial kernel and NN provide an average RMSE comparable to that of the linear counterpart.  The two methods also show a similar level of the robustness of the estimator (through standard deviation).  The proposed randomized subspace modeling achieves the best performance among all the methods.  The model based on A1 GCV criterion is comparable to the best linear methods, and the model based on A2 outperforms all other methods with respect to the average and standard deviation of the RMSE's. This indicates that the proposed methods provide a model that is not only more accurate but also more robust and stable while reducing uncertainty. 

\begin{table}[h!]
\caption{Comparison of the average and standard deviation of RMSE's from 5-fold evaluation} 
\centering 
\begin{tabular}{c c c} 
\toprule
Method & Average & Standard deviation \\ 
\midrule 
Linear regression & 12.61 & 1.05  \\ 
Lasso regression & 12.02 & 1.45  \\
Ridge regression & 11.98 & 1.20  \\
Principal component regression & 12.21 & 1.64  \\
Partial least squares regression & 11.89 & 1.15  \\ 
\midrule 
Random forest & 14.70 & 3.00  \\ 
$k$-nearest neighbors & 15.27 & 3.51  \\
SVM (RBF) & 17.49 & 3.77  \\
SVM (POLY) & 12.01 & 1.20  \\
Neural network &12.45&1.13\\
\midrule
Subspace-SVM (A1) & 11.99 & 1.12\\
\textbf{Subspace-SVM (A2)} & \textbf{11.64} &\textbf{0.94}\\
\bottomrule
\end{tabular}
\label{table:linear}
\end{table}

Although A1 applies the exact calculation of the GCV criterion, it is still based on the assumption of a squared loss function.  As such, A1 is also an approximation of the exact GCV criterion.  Since $\epsilon$-insensitive loss function used in Eq.~\eqref{eq:obj} penalizes errors out of the $\epsilon$-tube proportionally to the distance from the tube, the exact calculation for a squared loss function that more harshly penalizes larger errors may not guarantee the optimal performance.  In A2, simplifying the hat matrix, i.e., dropping the recursive term of $(\mathbf{I} - \sum_{l=0}^{j-1}{\mathbf{S}_l})$, can avoid dependency on other $\mathbf{S}_j$'s, so that it can reduce the variance of the estimator preventing potential overfitting and hence producing better results.  In addition, from the computational perspective, A2 does not require storing and adding all the previous $\mathbf{S}_j$'s but just requires calculating the trace of $\mathbf{S}_j^\prime$ for each subspace model estimation, so it is computationally more efficient than A1.

By construction, the randomized subspace model extracts critical subspaces.  However, none of the other methods has this capability.  To assess the variable selection capability of the randomized subspace method, we instead compare which individual variables are selected as an important predictor.  As illustrated in Fig.~\ref{fg:indvar}, we select common variables included in all five sets of critical subspaces from the 5-fold evaluation. For all other methods, we first find a set of important individual variables from each of the 5-fold evaluation and select the common variables from the sets for the consistency of the evaluation procedure.

\begin{figure}[h!]
\centering
\includegraphics[width=\textwidth]{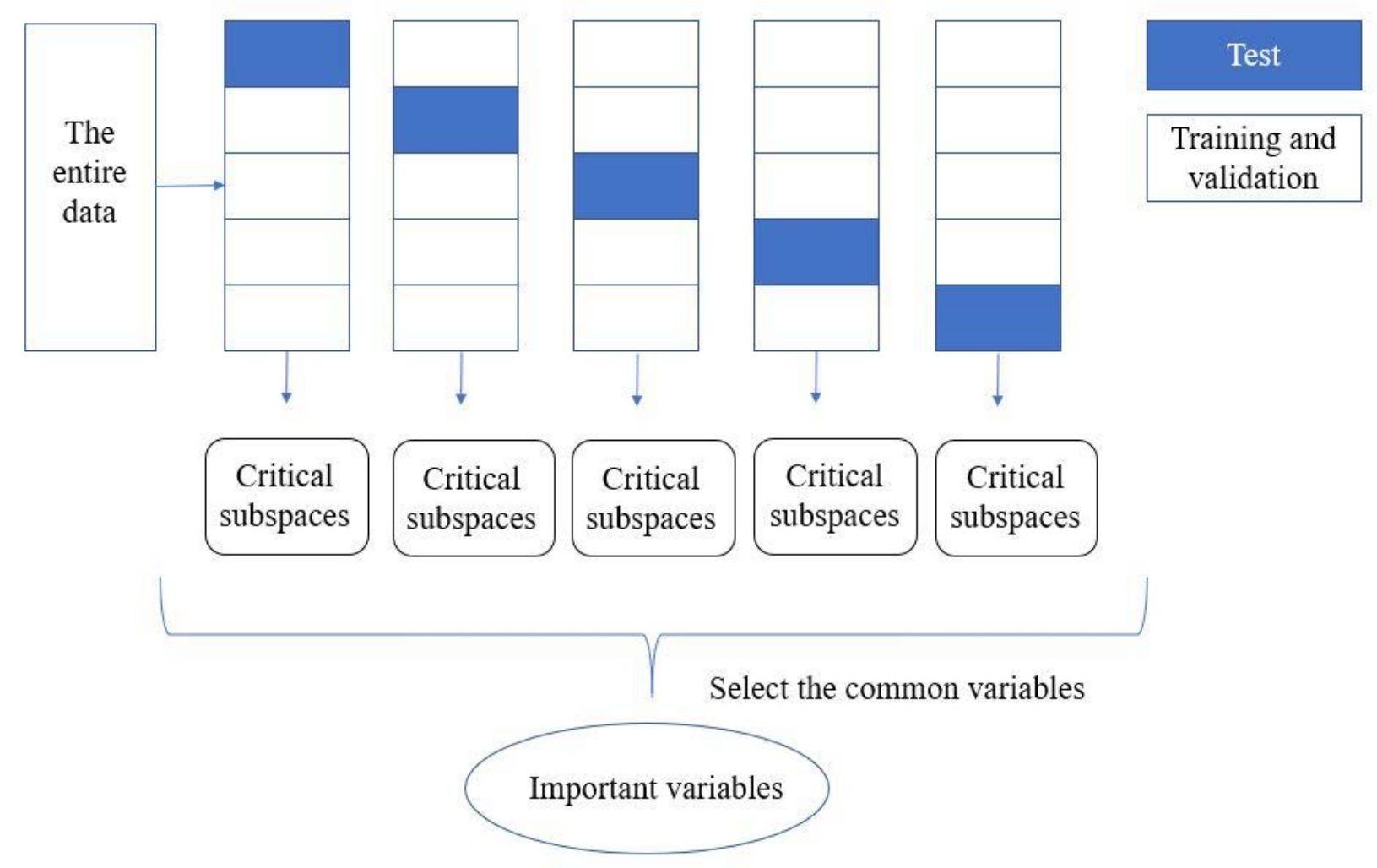}
\caption{Selection procedure of important individual variables}
\label{fg:indvar}
\end{figure}

\begin{table}[hbt!]
\caption{Important variables selected by all alternatives} 
\centering 
\resizebox{\textwidth}{!}{
\begin{tabular}{c|ccccccccccccccccccccccc } 
\toprule
Method & $x_1$ & $x_2$ & $x_3$& $x_4$& $x_5$& $x_6$& $x_7$& $x_8$& $x_9$& $x_{10}$ & $x_{11}$ & $x_{12}$ &$x_{13}$ &$x_{14}$& $x_{15}$ &$x_{16}$ &$x_{17}$& $x_{18}$ \\ 
\midrule 
Linear &\checkmark &\checkmark&\checkmark&\checkmark&\checkmark &\checkmark&\checkmark&\checkmark&\checkmark&\checkmark&&&&&&&&   \\ 
Lasso &\checkmark&\checkmark&\checkmark&\checkmark&\checkmark&\checkmark&\checkmark&\checkmark&&\checkmark&\checkmark&\checkmark&&&&&&  \\
Ridge&\checkmark&\checkmark&\checkmark&\checkmark&\checkmark &\checkmark&\checkmark &\checkmark&&&\checkmark& &\checkmark &\checkmark&&&&\\
PCR&\checkmark&\checkmark&\checkmark&\checkmark & & &\checkmark&&&&&\checkmark&\checkmark &\checkmark&\checkmark &\checkmark&& \\
PLS&\checkmark&\checkmark&\checkmark&\checkmark&\checkmark& &\checkmark&\checkmark&&\checkmark&&\checkmark&&&&&&\checkmark \\
RF&\checkmark&\checkmark&\checkmark&\checkmark&\checkmark&&\checkmark&\checkmark&&&&&\checkmark&&&\checkmark&\checkmark&\\
$k$-NN&\checkmark&\checkmark&\checkmark&\checkmark&&&\checkmark&\checkmark&&&&\checkmark&\checkmark&\checkmark&\checkmark&&\checkmark&\\
SVM (RBF)&\checkmark&\checkmark&\checkmark&\checkmark&&&\checkmark&\checkmark&&&&\checkmark&\checkmark&\checkmark&\checkmark&&\checkmark&\\
SVM(POLY)&\checkmark&\checkmark&\checkmark&\checkmark&&&\checkmark&\checkmark&&&&\checkmark&\checkmark&\checkmark&\checkmark&&\checkmark&\\
NN&\checkmark&\checkmark&&\checkmark&\checkmark&\checkmark&\checkmark&\checkmark&&&&\checkmark&&&&&&\\
Our method&\checkmark&\checkmark&\checkmark&\checkmark&\checkmark&\checkmark&\checkmark&\checkmark&&\checkmark&&\checkmark&&&&&&\\
\bottomrule
\end{tabular}}
\label{table:var} 
\end{table}

The result of important individual variable selection is shown in Table~\ref{table:var}.  In the table, certain variables are identified as important variables by all methods while some are selected by only a few methods.  For illustration purpose, we exclude all other variables that are not selected by any of the methods.  According to the variable selection result, PCR and NN missed a variable ($x_8$ and $x_3$, respectively) that is chosen by all others.  On the other hand, our method captures all variables that are found important by all the other methods, such as $x_1$, $x_2$, $x_4$ and $x_7$ as well as $x_8$ and $x_3$.  Variables identified by the majority of the other methods (e.g., $x_5$ and $x_{12}$) are marked as important in our method.  Certain variables rarely identified by the others (e.g., $x_9$ and $x_{11}$) are not included in the set of important variables for our method.  In a few cases, our method selects a variable not chosen by the majority (such as $x_6$ and $x_{10}$) and does not select a variable voted by the majority ($x_{13}$ and $x_{14}$).  Note here that, however, according to Table~\ref{table:linear}, the prediction quality of RF, $k$-NN, and SVM with an RBF kernel are not quite acceptable.  By excluding the variable selection results from these methods, our method is, in fact, capable of determining important and unimportant variables of which result well aligns with the results of the remaining methods.  In summary, this comparison result demonstrates that our method effectively extracts important variables.

Besides the popular machine learning methods, we also compare our variable selection result with that of the elementary effects method \citep{heng2018prediction} in Table~\ref{table:ifa}.  The result of the elementary effects method was derived from a dataset generated by one-factor-at-a-time approach, and the method did not extract common variables from 5-fold CV and did not record the top 20 important variables (instead, reported the top 10 important variables).  Albeit this is not a valid methodology comparison as there are other aforementioned factors affecting the variable selection, we aim to compare our finding to the existing variable selection for this specific problem.  In Table~\ref{table:ifa}, there are four common variables selected by our method and the elementary effects method. The elementary effects method does not mark some variables found important by the majority of the machine learning-based methods, including $x_2$, $x_3$, $x_7$, and $x_8$, and it identifies some variables that are never selected by others, such as $x_{24}$, $x_{31}$, $x_{33}$, and $x_{41}$.  Without knowing the true importance of the variables, we cannot draw a solid conclusion, but our methods result well aligns with the result of other alternatives whereas the existing selection from the elementary effects method is a bit far from the majority vote.  Our method adds $x_2$, $x_3$, $x_5$, $x_6$, $x_7$, and $x_8$ as important variables and drops $x_{11}$, $x_{17}$, $x_{24}$, $x_{31}$, $x_{33}$, and $x_{41}$ to and from the existing variable selection.

\begin{table}[hbt!]
\caption{Important variables selected by elementary effects method and our method} 
\centering 
\resizebox{\textwidth}{!}{
\begin{tabular}{c|ccccccccccccccccccc} 
\toprule
Method & $x_1$ & $x_2$ & $x_3$& $x_4$& $x_5$& $x_6$& $x_7$& $x_8$& $x_9$& $x_{10}$ & $x_{11}$ & $x_{12}$ &$x_{17}$ &$x_{24}$ & $x_{31}$ &$x_{33}$ &$x_{41}$\\ 
\midrule 
Elementary effects method &\checkmark&&&\checkmark&&& &&&\checkmark&\checkmark&\checkmark&\checkmark &\checkmark&\checkmark&\checkmark&\checkmark\\

Our method&\checkmark&\checkmark&\checkmark&\checkmark&\checkmark&\checkmark&\checkmark&\checkmark&&\checkmark&&\checkmark&&&&&&\\
\bottomrule
\end{tabular}}
\label{table:ifa} 
\end{table}

\subsection {Critical Subspace Analysis}
One of the major novelty of the randomized subspace modeling is the capability of extracting critical subspaces, equivalently, identifying important physical variables and important interactions among the variables.  In Section~\ref{ssc:result}, the comparison result demonstrates the quality of the important variable selection of the proposed method.  Note here that the result of the important (individual) variable selection was derived from critical subspaces identified in each of model learning/testing data splits.  This ensures the quality of the critical subspace selection to some extent.

To determine critical subspaces for the given dataset, we apply the randomized subspace modeling to the entire data, i.e., using all 200 observations, without any data split to learn the model.  The final model is constructed by 8 distinct critical subspaces, as shown in Table~\ref{table:subspace}.  We observe that these subspaces contain most of the important individual variables found by our method in Table~\ref{table:var}, except $x_5$, $x_6$, and $x_{12}$, while each subspace includes at least one or two important individual variable(s).  Assuming that an interaction formed by individually important variables is more significant than an interaction between others, we can at least conclude that the interactions between $x_2$ and $x_4$, between $x_4$ and $x_{10}$, and between $x_8$ and $x_{10}$ are significant.  

There are some variables that are not identified as important individual variables but included in the critical subspaces.  These variables by themselves may not have significant impact on the damage tolerance, but their interactions with others could influence the response.  Still, in its current form, we cannot validate if any combination between all variables in a subspace is significant.  We will investigate this aspect as one of future study.  For a reference, if only one variable is added to the model at a time ($k=1$), the prediction error is 12.28 with the standard error of 0.84.  By ignoring the interactions formed by multiple variables, the prediction capability becomes weaker.  This supports, at least to the minimum extent, that some of the interactions we modeled through subspaces improve the prediction, and hence they are considered significant.
\begin{table}[t!]
\caption{Critical subspace selected by our method} 
\centering 
\begin{tabular}{c c c} 
\toprule
Critical subspace & Important individual variables\\ 
\midrule 
$x_{19}, x_7, x_{20}$ & $x_7$   \\
$x_3, x_{15}, x_{21}$ & $x_3$  \\
$x_{22}, x_1, x_{16}$ & $x_1$  \\
$x_{23}, x_4, x_{22}$ & $x_4$  \\
$x_{15}, x_4, x_{10}$ & $x_4$, $x_{10}$  \\
$x_{24}, x_3, x_{25}$ & $x_3$  \\
$x_4, x_{17}, x_2$ & $x_4, x_2$ \\
$x_8, x_{10}, x_{26}$ & $x_8$, $x_{10}$ \\
\bottomrule
\end{tabular}
\label{table:subspace}
\end{table}

\section{Conclusion}
\label{sc:conc}
In this paper, we propose the randomized subspace modeling to alleviate challenges in analyzing high-dimensional datasets and provide valuable insight about an underlying system by identifying important physical variables and interactions. The proposed method leverages an ensemble of multiple models derived from critical subspaces, reduced-dimensional physical spaces. The critical subspaces are generated by a randomized search and evaluated by a cross-validated selection criterion.  With this structure, the proposed method shows its superiority over other alternatives in modeling a regression problem and identifying important variables.  Specifically, for a high-dimensional analysis of hybrid material's damage tolerance, we can draw the following conclusions:
\begin{enumerate}
    \item Compared to the models commonly used in the literature of variable selection (e.g., RF, NN, and Lasso), our method produces the lowest prediction error for this high dimensional, complex problem, demonstrating that our method is competitive for predicting high-dimensional datasets.  In addition, the lowest standard error shows that our method produces robust prediction.
    \item The result of our method's important variable selection well aligns with the majority of other alternatives, verifying our method's variable selection capability.  More importantly, our method identifies critical subspaces capturing not only important physical variables but also significant interactions, which is beneficial to experimental designs for broad engineering problems.
\end{enumerate}

In the future, we plan to improve the randomized search process.  This includes a weighted randomized search leveraging variable importance priors obtained from domain knowledge and a multi-step search process that first learns the rankings of variable importance and searches subspaces based on the rankings. Improving the randomized search can avoid adding insignificant variables as a part of the final subspace selection. However, the inherent randomness can still involve an insignificant variable in a subspace.  As such, we will also study how to identify and drop insignificant variables from the chosen subspaces.

\bibliographystyle{plainnat}
\bibliography{ref}

\end{document}